\documentclass[12pt,a4paper]{article}

\pdfoutput=1

\usepackage[margin=1in]{geometry}

\usepackage[
  bookmarks=false, 
  bookmarksnumbered=true,
  bookmarksopen=true,
  pdfborder={0 0 0},
  breaklinks=true,
  colorlinks=true,
  linkcolor=black,
  citecolor=black,
  filecolor=black,
  urlcolor=black,
]{hyperref}
\usepackage[round]{natbib}
\usepackage{amsmath,amsthm,amssymb}
\usepackage{graphicx}
\usepackage{enumitem}
\usepackage[capitalize]{cleveref}
\usepackage{ifthen}

\usepackage{xpatch}

\usepackage{nicefrac}
\usepackage{mathtools}

\newlist{equilibria}{enumerate}{1}
\Crefname{equilibriai}{Equilibrium}{Equilibria}
\setlist[equilibria,1]{label=(\arabic*),ref=(\arabic*)}

\Crefname{equation}{Equation}{Equations}

\Crefname{footnote}{footnote}{footnotes}

\theoremstyle{plain}
\newtheorem{theorem}{Theorem}[section]
\newtheorem{lemma}[theorem]{Lemma}
\newtheorem{proposition}[theorem]{Proposition}

\theoremstyle{definition}
\newtheorem{example}[theorem]{Example}
\newtheorem{remark}[theorem]{Remark}

\usepackage{amsmath,amsfonts,amssymb,bbm}

\newcommand{\I}{I}
\renewcommand{\S}{S}

\hypersetup{
  pdfauthor      = {Moshe Babaioff <moshe@microsoft.com>, Yannai A. Gonczarowski <yannai@gonch.name>, and Assaf Romm <assaf.romm@mail.huji.ac.il>},
  pdftitle       = {Playing on a Level Field: Sincere and Sophisticated Players in the Boston Mechanism with a Coarse Priority Structure},
}
\begin{document}
\title{Playing on a Level Field:\texorpdfstring{\\}{ }Sincere and Sophisticated Players in the Boston Mechanism with a Coarse Priority Structure}
\author{Moshe Babaioff \and Yannai A. Gonczarowski \and Assaf Romm\thanks{First draft: February 2018.
Babaioff: Microsoft Research, \emph{email}: \href{mailto:moshe@microsoft.com}{moshe@microsoft.com}.
Gonczarowski: Microsoft Research, \emph{email}: \href{mailto:yannai@gonch.name}{yannai@gonch.name}; the research was conducted while Gonczarowski was affiliated also with the Hebrew University of Jerusalem.
Romm: Stanford University and the Hebrew University of Jerusalem, \emph{email}: \href{mailto:assafr@gmail.com}{assafr@gmail.com}; the research was conducted while Romm was co-affiliated with Microsoft Research.
This paper greatly benefited from discussions with \mbox{Scott Duke Kominers}, \mbox{Tal Lancewicki}, \mbox{D\'eborah Marciano}, \mbox{Alvin Roth}, and \mbox{Ran Shorrer}. \mbox{Yannai Gonczarowski} was supported by the Adams Fellowship Program of the Israel Academy of Sciences and Humanities; his work was supported by Israel Science Foundation (ISF) grant 1435/14 administered by the Israeli Academy of Sciences, by United States-Israel Binational Science Foundation (BSF) grant 2014389, and by the European Research Council (ERC) under the auspices of the European Union's Horizon 2020 research and innovation program (grant 740282). The work of \mbox{Assaf Romm} is supported by the Maurice Falk Institute, by Israel Science Foundation (ISF) grant 1780/16, by United States-Israel Binational Science Foundation (BSF) grant 2016015, and by Koret Young Israeli Scholars Program.}}
\date{June 9, 2020}

\maketitle

\begin{abstract}
Who gains and who loses from a manipulable school-choice mechanism? Studying the outcomes of sincere and sophisticated students under the manipulable Boston Mechanism as compared with the strategy-proof Deferred Acceptance, we provide robust ``anything-goes'' theorems for large random markets with coarse priority structures. I.e., there are many sincere and sophisticated students who prefer the Boston Mechanism to Deferred Acceptance, and vice versa. Some populations may even benefit from being sincere (if also perceived as such). Our findings reconcile qualitative differences between previous theory and known empirical results. We conclude by studying market forces that can influence the choice between these mechanisms.
\end{abstract}


\section{Introduction} \label{sec:introduction}
School districts all over the world are increasingly realizing the benefits of letting parents choose the educational environment that best fits their children. The move from a zoning policy to a choice-oriented process, together with the scarcity of seats in some of the most highly demanded schools, calls for a regulated procedure or an algorithm for assigning seats. The Boston mechanism (henceforth BM), also known as Immediate Acceptance, is one of the most popular mechanisms for seat allocation. The appeal of BM lies in its simplicity in terms of both intuition and implementation. The mechanism, which also takes the parents' preferences into account, first maximizes the number of students who get their first choice, then maximizes the number of students who get their second choice, and so on. When performing this step-by-step maximization, the mechanism selects who gets admitted to overdemanded schools according to some (possibly school-specific) priority ordering. This process is so simple that in some small municipalities it is carried out manually using spreadsheet software.

Despite its seemingly straightforward description, the drawback of BM is that it is susceptible to strategic manipulation. Students (or parents) who carefully consider the workings of the mechanism can submit a rank-order list (ROL) that does not represent their true preferences, but may help them to be admitted to one of their (true) top schools. This raises two closely related concerns: first, do students benefit or are they harmed from using a manipulable mechanism and, second, do students who are more sophisticated and more informed get an unfair advantage in public-school admissions? A good way to measure gain and loss in this context is to compare the outcome of BM with that of the student-proposing Deferred Acceptance algorithm (\citealp{gs1962}; henceforth DA), which is strategy-proof and therefore does not give an advantage to sophisticated students \citep{df1981,roth1982}.\footnote{The comparison of BM to DA is appealing also because BM was replaced by DA in several school choice systems that were redesigned by economists, such as in Boston \citep{aprs2005} and New York City \citep{apr2009}.} We say that a student receives a positive \emph{absolute gain} (from BM) if her expected utility under (an equilibrium of) BM is higher than her utility under (the dominant-strategy equilibrium of) DA. We say that a student receives a \emph{relative gain} with respect to another student (or to a counterfactual version of the same student) if the absolute gain of the former is higher. Conveniently, a comparison of the relative gains between two types of the same student who differ only in their level of sophistication reduces to a comparison of their expected utilities under BM, since their outcome under DA is the same.

The first to approach these questions were \citet{ps2008}. They show positive absolute gain for sophisticated students, by demonstrating that a sophisticated student weakly prefers her outcome under the Pareto-dominant equilibrium of BM to her outcome under DA.\footnote{For sincere students the comparison is ambiguous.} They then prove positive relative gains for sophisticated types compared with sincere types of the same player. That is, holding everything else fixed and focusing on the Pareto-dominant equilibria of BM, a player is weakly better off being sophisticated rather than sincere.
	
It should be noted that both of these results require that schools have strict priority orderings over students---an assumption that is often unrealistic. In most cases, schools have a coarse priority structure and ties are resolved using a random tie-breaking rule. For example, children who have siblings attending a specific school may be given higher priority, but among children with such siblings there is no strict order, nor is there one among children who do not have siblings at that school. Allowing for a coarse priority structure, \citet{acy2011} first prove strictly positive absolute gains when all players are sophisticated, a result that is driven by the potential of BM to better express cardinal utility levels through preference reports.\footnote{\citet{miralles2009} states a similar result in a model with no sincere students.} \citeauthor{acy2011} then show that in the presence of both sophisticated and sincere students, sincere students may potentially experience positive absolute gains as well, as they benefit from their strategic peers' demand shading under BM, and gain higher probability of being admitted to one of their top schools (compared with the DA outcome).\footnote{\citet{ab2012} further demonstrate that positive absolute gains are possible for sincere students even prior to knowing their utility levels from being assigned to different schools.} That being said, in their model it is always better to be sophisticated (by assumption, since sophistication levels are random and unobserved), and thus their model always predicts nonnegative relative gains for sophisticated types compared with sincere types. Their results crucially rely on all players having common ordinal preferences, and on specific symmetry assumptions: in their model, cardinal utility levels are drawn i.i.d.\ from a random distribution, players are sincere or sophisticated with equal probability, and the results apply only to symmetric equilibria. As noted by \citet{troyan2012}, the welfare comparisons of \citeauthor{acy2011} are also sensitive to the assumption that there is only a single priority class, an assumption that mandates symmetric tie-breaking.

This paper complements, extends, and hopefully clarifies existing results on the absolute and relative gains of sincere and sophisticated students. Regarding relative gains, we observe that under weak priorities, being sincere can sometimes be an advantage rather than a liability. In fact, it is a likely situation for many students even in large random markets.\footnote{We study markets that are ``large'' in the sense that there are at the same time many schools and many students \citep[as in][]{im2005,kp2009}. A different model used in the literature is one with many students (a continuum) but a fixed number of schools \citep{al2016}. This paper does not attempt to study the latter model.} The intuition is that a sincere student commits to list her top schools first, even if they are overdemanded, and in doing so she crowds out the competition more than a sophisticated student would. That is, other students who are sophisticated inevitably take the sincere student's commitment into account and in doing so they rationally move toward ranking other schools higher. The question of whether a student is better off being sincere or sophisticated is thus similar to the decision faced by a Stackelberg leader in a meta-game in which she can either commit to one specific strategy (be sincere) or not commit to any strategy and play according to a Nash equilibrium profile (be sophisticated).

Our first main result, \cref{some-sophisticated-then-advantage-and-disadvantage}, demonstrates that the first-mover advantage of sincere students is prominent even in large random markets. Our model employs a random market-generating process that produces markets with heterogeneous preferences and varying demands by mixtures of sophisticated and sincere students. We show that there is in expectation a constant fraction of students who strictly prefer to be sincere.\footnote{Perhaps less surprisingly there is also a constant fraction of students who strictly prefer to be sophisticated.} Nevertheless, the sincerity advantage under BM is sensitive both to market structure and to the extent that the players' sophistication is common knowledge. Roughly put, the overall effect of being sincere as compared with being sophisticated is comprised both of the negative effect of not responding to excess demand (identified by \citealp{ps2008}), which manifests in both strict-priority and weak-priority environments, and of the positive effect of crowding out the competition, which is completely absent in strict-priority environments. For example, if two or more of each student's top schools are likely to be overdemanded, sincerity will allow a student to crowd out the competition in her top school. However, a sincere student who does not get her first choice pays an implicit cost of not having ranked her second choice first, which could potentially have represented a better trade-off between utility and admission probability. This nonresponsiveness to the excess demand of the top-choice school may or may not overshadow the benefit of the crowding-out effect.

Importantly, this result should not be narrowly interpreted. First, while the result in the main text is formulated with complete information on students' sophistication levels, it in fact holds even in the presence of incomplete information (see \cref{apndx:incomplete}). Second, we view this result as a statement of the relative gains not of individual students, but of distinguishable populations. In the real world, only some of each student's traits (e.g., gender, neighborhood, socioeconomic status) are observable by others. Sophisticated students who face a student with certain traits should evaluate her probability of being sincere based on prior information on all students with similar traits. Therefore, while it is (almost tautologically) better to be sophisticated when holding everything else (including others' beliefs regarding oneself) fixed, it is possible that an entire population will benefit from most or even all of its members being sincere, provided that the sophisticated players know that the population is such. This is true to the extent that the sincere members of such a mixed population may be better off than the sophisticated members of a pure population. It is important to note that another possible reason for wanting to belong to a sincere population is that in real markets one may often face members of one's own community, and their sincerity may play to one's advantage. In order to abstract away from this effect of benefiting from the actions of others whose sincerity level is correlated with one's own, we assume in the main text that each population is represented by a single student. This ensures that our results are not driven by any interaction between two students of the same population. We further simplify the arguments by treating sophistication as a binary variable. \cref{apndx:incomplete} provides details on relaxing these restrictions.

Following our study of the relative gains, we turn to study the absolute gains of each type of student. Our second main result compares equilibria of BM with the DA outcome and shows that while there are indeed many students who prefer DA to BM, there are also many students who have the opposite preference, and that this prediction holds true for both sophisticated and sincere students. The key insight here is that, compared with DA, BM~may reduce students' opportunities to compete for desirable (overdemanded) schools. This effect, which is once again absent in strict-priority environments, may in fact overcome the reduced-competition effect that \citet{ps2008} identify. We show that in large random markets there is in expectation a constant fraction of students (either sincere or sophisticated) who strictly prefer the DA outcome to any BM equilibrium, as well as a constant fraction of students who have the opposite strict preference.

An important aspect of our work is that it can explain the very different empirical results that have been documented and analyzed in recent years. Studying elementary-school choice in Cambridge, MA, \citet{as2018} show absolute gains for both sophisticated and sincere players, with sophisticated students gaining more relative to sincere students. \citet{cfg2018} present qualitatively similar results from elementary-school choice in Barcelona. \citet{he2017} uses data on middle-school choices in a neighborhood in Beijing and that is the only paper that we know of that finds relative gains for sincere students compared with sophisticated students. \citet{knz2020} do not explicitly separate agents into types, but argue that subjective (and possibly inaccurate) beliefs about demand could lead to negative absolute gains. \Citet{dgov2015} use data from Amsterdam to show negative absolute gains (i.e., students get higher expected utilities under DA).

We note that all of our arguments and results apply also beyond the plain vanilla BM, and in fact hold for virtually all known variants of BM: the Corrected Boston Mechanism \citep{miralles2009}, the Modified Boston Mechanism \citep{dur2019}, the Secure Boston Mechanism \citep{dhm2019}, and the Adaptive Boston Mechanism \citep{ms2017}. In fact, they even hold for the entire family of First-Choice Maximal mechanisms \citep{dms2018}. Our proofs do not directly apply to manipulable mechanisms outside this family, such as DA with bounded lists (when the bound is binding and prevents agents from playing truthfully).

The paper proceeds as follows. After discussing related works, we present the model in \cref{sec:model}, the results on relative gains in \cref{sec:sincere}, and the results on absolute gains in \cref{sec:bm-da}. A discussion of the applicability of the results to realistic environments follows in \cref{sec:what_matters}, and \cref{sec:discussion} concludes. Some extensions are presented in \cref{apndx:incomplete}, and some proofs and calculations are relegated to \cref{apndx:proofs,apndx:sincerity_bad}.
	
\subsection{Additional Related Literature} \label{subsec:lit_review}	
The mechanism-design approach to school choice begins with \citet{as2003}. Since then much of the academic literature has focused on the application of DA to various environments. BM itself was brought to the attention of economists by \citet{aprs2005}, who redesigned Boston Public Schools' existing mechanism and replaced it with~DA.
	
\citet{aprs2005} noted that BM was prone to strategic manipulation, and pointed to anecdotal evidence suggesting that indeed some parents in Boston acted strategically. Following that, strategic behavior was demonstrated in both experimental labs (e.g., \citealp{cs2006}, and many follow-up designs) and real-world environments (see, e.g., \citealp{cg2014}). \citet{es2006} show that when priorities are strict and all students are sophisticated, moving from BM to DA is weakly beneficial for all students.\footnote{See also \citet{kojima2008} for similar results under substitutable priority structures.} As mentioned, \citet{ps2008} generalize this statement, claiming that the strategic manipulability of BM also has fairness implications, as it gives an advantage to sophisticated students. An experimental finding along the same lines was recently presented by \citet{bm2016}. Similar (if inconclusive) theoretical results for an environment with boundedly sophisticated (level-$k$) players are described by \citet{zhang2016}. \citet{dhm2018}\ study strategic behavior of students and its implications in the field. \citet{av2018} note that acting sincerely is not necessarily related to being strategically unsophisticated or suffering from lack of information; it can also be the result of having better outside options.

Strategic sophistication (or lack thereof) and its effect on designing matching markets has recently been at the focus of a number of works that study preference misrepresentation under the strategy-proof DA. Experimental evidence for preference misrepresentation have existed for more than a decade, starting with the pioneering work of \citep{cs2006}. More recently, \citet{rees2018} and \citet{rs2018} studied this phenomenon in the American market for new medical residents, and a few other authors pointed at field evidence for preference misrepresentation in college admission markets \citep{ach2017,hrs2020,ss2018}. Some explanations for this behavior are suggested by \citet{hmrs2017}, and more formal treatments of some of them are provided by \citet{ag2018} and \citet{dhr2019}.

Finally, this paper deals with the effects of weak priorities on the workings of BM. It thus contributes to a recent line of works that deal with weak priorities and tie-breaking methods for other school choice mechanisms, such as DA (e.g., \citealp{apr2009,arnosti2015,an2016,anr2019,ee2008,kesten2010}) and Top Trading Cycles (e.g., \citealp{as1998,carroll2014,ps2011}).

\section{The Model} \label{sec:model}	

\paragraph{Schools, Students, and Preferences}
We adopt much of the notation used by \citet{ps2008}. There is a finite set of \emph{schools}, $\S = \{s_1,\dots,s_m\}$ and a finite set of \emph{students}, $\I = \{i_1,\dots,i_n\}$. Each school $s_j$ has a \emph{capacity} $q_{s_j}$. The (vNM) \emph{utility} of a student $i\in\I$ from being assigned to school $s\in\S$ is~$u_i(s)$, the utility from being unmatched is (normalized to) zero, and we assume that students are risk-neutral. Throughout the paper, we assume that the \emph{priority structure} of the schools is coarse to the extent that all students belong in the same priority class in all schools. While this is a simplifying assumption, it is not far from the properties of many real-life matching markets, and the resulting phenomena efficiently convey our main messages.

\paragraph{The Boston Mechanism} Each student must report an ROL to BM, and then the mechanism is run as follows:
\begin{enumerate}
\item[0.]
For each school $s$, a strict ordering of students $\tau_s$ is drawn uniformly at random from the set of all permutations over\footnote{\label{mtb-stb}This is the \emph{multiple tie-breaking} rule (MTB). However, all our results remain qualitatively the same for other random tie-breaking rules, including the widely used \emph{single tie-breaking} rule (STB) under which all schools use the same random ordering of students.} $\I$. This ordering remains unknown to the students.
\item
Each student applies to the school that she ranked as her first choice. A school whose capacity is at least the number of students who applied \emph{permanently} admits all of them. A school $s$ whose capacity is less than the number of students who applied \emph{permanently} fills its capacity with a subset of these applicants, who are chosen according to $\tau_s$.
\item
Each student who was not admitted in the first round applies to her second choice. A school whose \emph{remaining capacity} (taking into account the slots taken by all of the students admitted in the first round) is at least the number of students who applied in this round \emph{permanently} admits all of them. A school $s$ whose remaining capacity is less than the number of students who applied in this round \emph{permanently} fills its remaining capacity with a subset of these applicants, who are chosen according to $\tau_s$.
\end{enumerate}
\begin{minipage}{\textwidth}\vspace{-.25em}\centering$\vdots$\vspace{.25em}\end{minipage}
\begin{enumerate}
\item[$k$.]
Each student who was not admitted in the previous rounds applies to her $k$th choice. A school whose remaining capacity (taking into account the slots taken by all of the students admitted in all previous rounds) is at least the number of students who applied in this round \emph{permanently} admits all of them. A school $s$ whose remaining capacity is less than the number of students who applied in this round \emph{permanently} fills its remaining capacity with a subset of these applicants, who are chosen according to $\tau_s$.
\end{enumerate}
\begin{minipage}{\textwidth}\vspace{-.25em}\centering$\vdots$\vspace{.25em}\end{minipage}

\paragraph{The Deferred-Acceptance Mechanism} Each student must report an ROL to DA, and then the mechanism is run as follows:
\begin{enumerate}
\item[0.]
For each school $s$, a strict ordering of students $\tau_s$ is drawn uniformly at random from the set of all permutations over\footnote{See \cref{mtb-stb}.} $\I$. This ordering remains unknown to the students.
\item
Each student applies to the school that she ranked as her first choice. A school whose capacity is at least the number of students who applied \emph{tentatively} admits all of them. A school $s$ whose capacity is less than the number of students who applied \emph{tentatively} admits a subset of these applicants who fill its capacity and who are chosen according to~$\tau_s$, and \emph{permanently} rejects all other applicants.
\end{enumerate}
\begin{minipage}{\textwidth}\vspace{-.25em}\centering$\vdots$\vspace{.25em}\end{minipage}
\begin{enumerate}
\item[$k$.]
Each student applies to her favorite school among those that have not rejected her yet. (Thus a student who was tentatively admitted to a school in the previous round reapplies to the same school in this round.) A school whose capacity is at least the number of students who apply in this round \emph{tentatively} admits all of them. A school $s$ whose capacity is less than the number of students who apply in this round \emph{tentatively} admits a subset of these applicants who fill its capacity and who are chosen according to $\tau_s$, and \emph{permanently} rejects all other applicants.
\end{enumerate}
\begin{minipage}{\textwidth}\vspace{-.25em}\centering$\vdots$\vspace{.25em}\end{minipage}

\vspace{\topsep}
The mechanism terminates when a round with no rejections occurs, following which all tentative admissions from this round become permanent. The resulting outcome is the student-optimal stable outcome.

\paragraph{Sincerity and Sophistication} Students are either sincere or sophisticated. \emph{Sincere} students truthfully report an ROL according to their utilities, while \emph{sophisticated} students can submit any ROL regardless of their utilities. When analyzing the Boston mechanism, we look at a Nash equilibrium of the preference-reporting game among the sophisticated students. (We assume truthful reporting under DA, as this constitutes a dominant-strategy equilibrium.)

\paragraph{Utilities}
While some of the examples in this paper use specific given utilities for each student, our main results apply to the \emph{uniform $(n;u^1,\ldots,u^k)$ model}, which we now describe. In this model, which is defined by fixed utilities $u^1>\cdots>u^k$ and a size $n$, there are $n$ students and $n$ schools, where each of the schools has a capacity of exactly $1$. For each student, we draw ordinal preferences uniformly at random, and set the utility for this student from being matched to a school that she ranks in place $\ell$ to be $u^\ell$. Formally, for each student $i\in\I$ we independently and uniformly draw $k$ distinct schools $s_{\pi_1},\ldots,s_{\pi_k}$, and set $u_{i}(s_{\pi_1})=u^1$, $u_{i}(s_{\pi_2})=u^2$, etc., where $u^\ell=0$ for all $\ell>k$. The short-list assumption is mostly for technical convenience, and most results and proofs are easily adapted to the case of unbounded lists. Specifically, we do \emph{not} rely on having many vacant schools in the market or anything of that nature. The sole exception is \cref{some-sophisticated-then-bm-da}, whose first part requires the length of the ROLs to be bounded by some arbitrary fixed number.

While utilities are randomly generated, we focus on games with complete information regarding players' (realized) preferences and sophistication levels.\footnote{In \cref{apndx:incomplete} we do analyze a more general model with incomplete information on sophistication levels.} The only uncertainty that players face is due to the random tie-breaking rule used in both mechanisms. 

\section{Relative Gains: Sincere vs.\ Sophisticated}\label{sec:sincere}
Our first main result speaks to the indeterminacy of the relative gains of sophisticated students as compared with sincere students under BM. As mentioned, under DA both types get the same utility. Therefore, we only need to verify how many students prefer to be sophisticated and how many prefer to be sincere under BM. 

We show that in large random markets there are many students who prefer being sophisticated, and many who prefer being sincere. This is thus an ``anything goes'' kind of theorem that stands in sharp contrast to the case of strict priorities, where, as \citet{ps2008} show, each student weakly prefers being sophisticated. Formally, we prove that in large random markets, for a large family of cardinal utility levels, and for any nontrivial proportion of sincere and sophisticated students, the expected number of students who prefer being sincere is linear in $n$, as is the expected number of students who prefer being sophisticated.

\begin{theorem}[Relative Gains are Often Positive for Some and Negative for Others]
\label{some-sophisticated-then-advantage-and-disadvantage}
Let $k=2$ and let $u_1>u_2 > \frac{2}{3}\cdot u_1 > 0$. For any $0 < p < 1$, there exists $\tau>0$ such that for any large enough $n$, when each student is sincere with probability $p$ independently of other students, in the uniform $(n; u^1,u^2)$ model both of the following hold:
		\begin{enumerate}
			\item\label[part]{some-sophisticated-then-advantage-and-disadvantage-sincere-prefer-sophisticated}
			There exists a set of sincere students of expected\footnote{The expectation is taken over preferences and sophistication levels.} size at least $\tau n$, such that each student in this set strictly prefers any equilibrium had she been sophisticated to any equilibrium (in which she is sincere). Furthermore, each student in this set maintains this strict preference regardless of whether or not any other students become sophisticated and/or sincere.
			\item
			For any equilibrium (in which only the sophisticated students strategize), there exists a set of sophisticated students of expected\footnote{The expectation is taken over preferences and sophistication levels (for any given mapping from realized preferences and sophistication levels to equilibria).} size at least $\tau n$, such that each student in this set strictly prefers any equilibrium had she been sincere to the given equilibrium. Furthermore, each student in this set maintains this strict preference regardless of whether or not any other students in this set become sincere and regardless of whether or not any sincere students become sophisticated.
		\end{enumerate}
	\end{theorem}
	
\begin{remark} \label[remark]{rmk:p01}
For $p = 0$ only the second part of the \lcnamecref{some-sophisticated-then-advantage-and-disadvantage} holds. For $p = 1$ only the first part of the \lcnamecref{some-sophisticated-then-advantage-and-disadvantage} holds.
\end{remark}	
\begin{remark} \label[remark]{rmk:fixed_utilities}
While we use a fixed cardinal utility form, the exact values are qualitatively inconsequential. What is important is that $u_1$ and $u_2$ are not too far apart, as otherwise all students would always rank their top school first.
\end{remark}
\begin{remark} \label[remark]{rmk:kgeq3}
It is straightforward to generalize this result and all other theorems and lemmas (except for the first part of \cref{some-sophisticated-then-bm-da}, as mentioned above) to arbitrary values of $k$ and even to preferences of unbounded length. One easy way to achieve this is to replace the current restrictions on the cardinal form with restrictions that mandate a significant decrease in students' utility when they receive their third alternative or worse.
\end{remark}
	
\bigskip

Before presenting the proof, we provide some intuition using two \lcnamecrefs{sincerity_good}. We first illustrate the basic idea using an \lcnamecref{sincerity_good} in which we restrict our focus to symmetric equilibria.
	
\begin{example}[Sincerity Can be an Advantage: Symmetric Equilibrium] \label[example]{sincerity_good}
Let $\S = \{s_1,s_2\}$ and let $\I = \{i_1,i_2,i_3\}$. Each $s \in \S$ has capacity $q_s = 1$. Every $i \in \I$ has $u_i(s_1) = 3$ and $u_i(s_2) = 2$. Suppose that $i_2$ and $i_3$ are sophisticated.
		
If $i_1$ is also sophisticated, the only \emph{symmetric} equilibrium is for each student to report $s_1\succ s_2$ with probability $0.8$ and to report $s_2\succ s_1$ with probability $0.2$. This leaves $i_1$ (as well as $i_2$ and $i_3$) with an expected utility of $\nicefrac{5}{3}$.
		
		If $i_1$ is sincere (and therefore reports $s_1\succ s_2$ with probability $1$), then the only equilibrium in which $i_2$ and $i_3$ play symmetrically is for each of them to report $s_1\succ s_2$ with probability~$0.6$ and to report $s_2\succ s_1$ with probability $0.4$ (each obtaining an expected utility of $\nicefrac{8}{5}<\nicefrac{5}{3}$). This leaves $s_1$ with an expected utility of $\nicefrac{9}{5}>\nicefrac{5}{3}$.
	\end{example}
	
The mechanics behind \cref{sincerity_good} are as follows: in a symmetric equilibrium among all students, since each student plays a mixed strategy, her expected utilities from each of the two pure strategies $s_1\succ s_2$ and $s_2 \succ s_1$ that she plays with positive probability are the same. By becoming sincere, $i_1$ in essence becomes a Stackelberg leader who commits to the strategy $s_1\succ s_2$ (instead of mixing it with $s_2 \succ s_1$), and by doing so, she ``crowds out'' the other (sophisticated) students in the school $s_1$, breaking the equality between the expected utilities of the other students from playing $s_1\succ s_2$ and $s_2 \succ s_1$, thereby shifting their mixed strategy in the new symmetric equilibrium toward playing $s_2 \succ s_1$ with higher probability and playing $s_1\succ s_2$ with lower probability. Since other students play $s_1\succ s_2$ with lower probability, the utility of $i_1$ from playing $s_1\succ s_2$ given the new strategies for all other students is higher than her utility from playing $s_1\succ s_2$ given the old strategies for all other students, which was her utility in the old equilibrium (since she played $s_1\succ s_2$ with positive probability in that equilibrium). Using this fundamental idea, we can generalize \cref{sincerity_good} to include more students and/or other utilities.
	
	While the market in \cref{sincerity_good} may seem to be very carefully crafted, e.g., in terms of alignment, \cref{all-sophisticated-then-prefer-sincere} below shows that the phenomenon identified in that \lcnamecref{sincerity_good} is in fact generic, and occurs often in a large uniform market. Before proceeding to that \lcnamecref{all-sophisticated-then-prefer-sincere}, we present a somewhat more complicated \lcnamecref{five-vs-six} where we do not restrict ourselves merely to symmetric equilibria.

	\begin{example}[Sincerity Can be an Advantage]\label[example]{five-vs-six}
		Let $\S = \{s_1,s_2,\ldots,s_6\}$ and $\I = \{i_1,i_2,\ldots,i_5\}$. Each $s \in \S$ has capacity $q_s = 1$. 
		Every $i \in \I$ has utility $4$ from being matched to her first choice, utility $3$ from being matched to her second choice, and utility $0$ otherwise. The preferences of the students are as follows:
		\begin{enumerate}
			\item
			$\succ_{i_1}: s_1,s_2$ (i.e., $i_1$ prefers $s_1$ first and $s_2$ second),
			\item
			$\succ_{i_2}: s_1,s_3$,
			\item
			$\succ_{i_3}: s_1,s_4$,
			\item
			$\succ_{i_4}: s_2,s_5$,
			\item
			$\succ_{i_5}: s_3,s_6$.
		\end{enumerate}
		Suppose that $i_1$ and $i_2$ are sophisticated. Regardless of whether $i_3$, $i_4$, and $i_5$ are sophisticated or not, we note that they will rank truthfully (since the second choice of each is guaranteed). Thus, there are three possible equilibria:\footnote{This example also demonstrates that a sincere student may be matched to different schools across different equilibria (here this can be observed for $i_4$ and for $i_5$). Furthermore, this holds even when restricting to pure-strategy Pareto-dominant equilibria, and even when looking only at utilities and not at the actual school to which the student is assigned. This phenomenon is in contrast to the case of strict priorities, where \citet[Proposition~2]{ps2008} show that each sincere student is matched to the same school across all equilibria.}
		\begin{equilibria}
			\item\label{eq-1-truthful}
			$i_1$ ranks truthfully, $i_2$ ranks $s_3$ at the top. In this case, the utility of $i_1$ is $\nicefrac{4}{2}=2$ and the utility of $i_2$ is $\nicefrac{3}{2}$.
			\item\label{eq-2-truthful}
			$i_2$ ranks truthfully, $i_1$ ranks $s_2$ at the top. In this case, the utility of $i_2$ is $\nicefrac{4}{2}=2$ and the utility of $i_1$ is $\nicefrac{3}{2}$.
			\item\label{eq-mix}
			$i_1$ and $i_2$ rank truthfully with probability $\nicefrac{3}{4}$ and rank their second choice at the top with probability $\nicefrac{1}{4}$.
			In this case, the utility of each of these students is $\nicefrac{3}{2}$.
		\end{equilibria}
		We note that \cref{eq-1-truthful} is strictly preferred by $i_1$ to the other two \lcnamecrefs{eq-2-truthful}, and is obtained if $i_1$ is sincere. Similarly, \cref{eq-2-truthful} is strictly preferred by $i_2$ to the other two \lcnamecrefs{eq-1-truthful}, and is obtained if $i_2$ is sincere. Thus, regardless of the \lcnamecref{eq-1-truthful}, at least one of the two students, $i_1$ or $i_2$, strictly prefers to become sincere.
	\end{example}
	
Our next result shows that in a large random market and under a broad range of cardinal utilities, \cref{five-vs-six} repeats linearly many times.
	
	\begin{lemma}[When Sophistication is Prevalent, Sincerity is Often an Advantage]\label[lemma]{all-sophisticated-then-prefer-sincere}
		Let $k=2$ and let $u_1 > u_2 > \frac{2}{3}\cdot u_1 > 0$. There exists a constant $\tau>0$ such that for any large enough $n$, when all students are sophisticated, in the uniform $(n; u^1,u^2)$ model the following holds: For any equilibrium (in which all of the students strategize), there exists a set of students of expected\footnote{The expectation is taken over preferences (for any given mapping from realized preferences to equilibria).} size at least $\tau n$, such that each student in this set strictly prefers any equilibrium had she been sincere\footnote{In such equilibria all students except herself strategize, while she is truthtelling.} to the given equilibrium. Furthermore, each student in this set maintains this strict preference regardless of whether or not any other students in this set become sincere.
	\end{lemma}
	
	\cref{all-sophisticated-then-prefer-sincere} analyzes a random market where all students are sophisticated. By contrast, in a large uniform market where all students are sincere, it is obviously weakly beneficial to become sophisticated. As \cref{all-sincere-then-prefer-sophisticated} shows, this is also strictly beneficial for linearly many students.
	
	\begin{lemma}[When Sincerity is Prevalent, it is Often a Disadvantage]
\label[lemma]{all-sincere-then-prefer-sophisticated}\leavevmode
		\begin{enumerate}
			\item\label[part]{all-sincere-then-prefer-sophisticated-weak}
			If all students are sincere, then each student weakly prefers to become sophisticated.
			\item\label[part]{all-sincere-then-prefer-sophisticated-strict}
			Let $k=2$ and let $u_1 > u_2 > \frac{2}{3}\cdot u_1 > 0$. There exists a constant $\tau>0$ such that for any large enough $n$, when all students are sincere, in the uniform $(n; u^1,u^2)$ model there exists a set of students of expected\footnote{The expectation is taken over preferences.} size at least $\tau n$, such that each student in this set strictly prefers any equilibrium had she been sophisticated\footnote{In such equilibria she strategizes, while all other students are truthtelling.} to the outcome (where all students are sincere). Furthermore, each student in this set maintains this strict preference (relative to equilibrium outcomes in which she is sincere) regardless of whether or not any other students become sophisticated.
		\end{enumerate}
	\end{lemma}
	
We are now ready to complete the proof of \cref{some-sophisticated-then-advantage-and-disadvantage}.

		\begin{proof}[Proof of \cref{some-sophisticated-then-advantage-and-disadvantage}]
		By a calculation similar to that in the proof of \cref{all-sincere-then-prefer-sophisticated}, the set of sincere students $z$ satisfying the conditions in that proof is of expected size at least $p\cdot\tau\cdot n$, where $\tau$ is as defined there, and so the first statement holds for $\tau_1=p\cdot\tau$ (recall that under the conditions in that proof, $y,x,w$ will rank truthfully even if they are sophisticated, since their second choice is guaranteed, and so we multiply $\tau$ by the $p$ probability of $z$ being sincere).
		By a calculation similar to that in the proof of \cref{all-sophisticated-then-prefer-sincere}/\cref{five-vs-six}, the second statement holds for $\tau_2=(1-p)^2\cdot\tau$, where $\tau$ is as defined in \cref{all-sophisticated-then-prefer-sincere} (once again, $x,w,v$ will rank truthfully regardless of whether or not they are sophisticated, since their second choice is guaranteed, and so we multiply by the $(1\!-\!p)^2$ probability of both $z$ and $y$ being sophisticated).
		The \lcnamecref{some-sophisticated-then-advantage-and-disadvantage} then holds with $\tau = \min\{\tau_1,\tau_2\}$.
	\end{proof}

\begin{remark} \label[remark]{rmk:boston_variants}
The fact that the proof uses only ROLs of length 2 immediately implies that the result also holds when BM is replaced by any First-Choice Maximal mechanism \citep{dms2018}, and in particular with the Corrected Boston Mechanism \citep{miralles2009}, the Modified Boston Mechanism \citep{dur2019}, and the Adaptive Boston Mechanism \citep{ms2017}. The result also holds under the Secure Boston Mechanism \citep{dhm2019} since the probability that a student is ranked first by the tie-breaking rule is only $\nicefrac{1}{n}$, and the proof remains almost exactly the same.
\end{remark}

\begin{remark} \label[remark]{rmk:populations_model}
It is possible to extend \cref{some-sophisticated-then-advantage-and-disadvantage} to a model with finitely many populations and with incomplete information on the sophistication levels of individual students (as long as the probability of being sophisticated is not too low). In this extension, each population contains a constant share of the students, and belonging to a population implies some population-specific probability of being either sincere or sophisticated. For details see \cref{apndx:incomplete}.
\end{remark}

	\section{Absolute Gains: Boston Mechanism vs.\ Deferred Acceptance}\label{sec:bm-da}
	
	In their paper, \citet{ps2008} compare, for sophisticated students, the outcomes of BM with those of DA (the student-optimal stable mechanism). They show that sophisticated students weakly prefer the Pareto-optimal equilibrium\footnote{When priorities are strict, a unique Pareto-optimal equilibrium indeed exists.} of BM to the student-optimal stable matching. Furthermore, they show that the set of all equilibria of BM in any given economy coincides with the set of all stable matchings in an ``augmented economy'' where all sincere students are demoted in the priorities of the schools that they do not rank first. These two results seem to suggest that BM is ``less fair'' than DA, which treats all students, whether sophisticated or sincere, equally. As we will now show, these phenomena do not necessarily continue to hold in the presence of weak priorities.
	
	\begin{proposition}[Sophisticated Students May Prefer DA to BM]\label[proposition]{boston-da}
		There exists a matching market with weak priorities
		that satisfies both of the following conditions:
		\begin{enumerate}
			\item
			The set of Nash equilibrium outcomes of BM in this matching market is different from the set of stable matchings under the ``augmented economy'' defined by \citet{ps2008} for this market. Furthermore, the union of the supports of the outcomes in the former set is different than in the latter set.
			\item
			BM has only one equilibrium (which is thus a Pareto-dominant equilibrium) in this market.
			The utility of each sophisticated student in the equilibrium of BM is strictly lower than her utility in the DA outcome.
			The utility of each sincere student in the equilibrium of BM is strictly higher than her utility in the DA outcome.
		\end{enumerate}
	\end{proposition}
	
	\begin{proof}
		Let $S=\{s_1,s_2\}$ and let $I=\{i_1,\ldots,i_{4}\}$. Every $s\in S$ has $q_s=1$. Every $i\in I$ prefers $s_1$ first, and receives utility $4$ from being matched to $s_1$. Students $i_1$ and $i_2$ are sincere and do not wish to be matched to $s_2$. Students $i_3$ and $i_4$ are sophisticated and each receive utility $3$ from being matched to $s_2$.
		Note that the augmented economy in this case is the same as the original economy.
		
		In the unique BM equilibrium, each of the sophisticated students goes to $s_2$ first, and so gets utility $\nicefrac{3}{2}$; each sincere student therefore gets utility $\nicefrac{4}{2}=2$ in this equilibrium. The union of the supports of all BM equilibrium outcomes is therefore the set of all matchings where some sincere student $i\in\{i_1,i_2\}$ is matched to $s_1$ and some sophisticated student $i\in\{i_3,i_4\}$ is matched to $s_2$.
		
		In the student-optimal stable mechanism, each sophisticated student goes to $s_1$ first, and so receives utility $\nicefrac{1}{4}\cdot4+\nicefrac{3}{4}\cdot(\nicefrac{1}{3}\cdot3+\nicefrac{2}{3}\cdot\nicefrac{3}{2})=\nicefrac{5}{2}>\nicefrac{3}{2}$;
		each sincere student therefore gets utility $\nicefrac{4}{4}=1<2$ in DA. The union of the supports of all stable matchings (in the augmented economy) is the set of all matchings where some (not necessarily sincere) student $i\in I$ is matched to $s_1$ and some sophisticated student $j\in\{i_3,i_4\}$ with $j\ne i$ is matched to~$s_2$.
	\end{proof}

	As in the study of the trade-off between sincerity and sophistication, one may ask whether, and to what extent, the above-demonstrated preference of sophisticated students for DA over BM remains prominent in a large random market. The second main result of this paper shows that, generally, both a constant fraction of any population prefers DA to BM and a constant fraction of any population prefers BM to DA, in expectation.

	\begin{theorem}[Absolute Gains are Often Positive for Some and Negative for Others]
\label{some-sophisticated-then-bm-da}
		Let $k=2$. There exists a constant $\tau>0$ such that for any large enough $n$,
		for any utilities $u^1>u^2>0$, and for any fixed assignment of students into sophisticated and sincere types, in the uniform $(n; u^1,u^2)$ model both of the following hold:
		\begin{enumerate}
			\item\label[part]{some-sophisticated-then-bm-da-prefer-bm}
			There exists a set consisting of an expected fraction of at least $\tau$ of the sophisticated students and an expected fraction of at least $\tau$ of the sincere students,\footnote{The expectation is taken over preferences.} such that each student in this set strictly prefers the DA outcome to any equilibrium of~BM.
			\item\label[part]{some-sophisticated-then-bm-da-prefer-da}
			There exists a set consisting of an expected fraction of at least $\tau$ of the sophisticated students and an expected fraction of at least $\tau$ of the sincere students,\footnote{Once again, the expectation is taken over preferences.} such that each student in this set strictly prefers any equilibrium of~BM to the DA outcome.
		\end{enumerate}
		Furthermore, each student in each of the above sets maintains this strict preference even if the assignment of students into sophisticated and sincere types changes arbitrarily.
	\end{theorem}
	
\begin{remark} \label[remark]{rmk:th2_all}
\cref{rmk:boston_variants,rmk:populations_model} also hold for \cref{some-sophisticated-then-bm-da}. Moreover, a stronger version of \cref{rmk:populations_model} actually holds, as there is no need to make any restrictions on the share of sophisticated students in the population.
\end{remark}

	\section{A Trade-Off between Positive and Negative Effects}\label{sec:what_matters}
	
	\subsection{The Implications of Sincerity:\texorpdfstring{\\}{ }Crowding Out Others vs.\ Not Responding to Excess Demand}\label{effects-sincerity}
	
	The results of \cref{sec:sincere} demonstrate a positive effect of the sincerity of a student $i$ in BM: the ability to crowd out others by essentially becoming a Stackelberg leader who commits to ranking in accordance with her true preference, forcing (other) sophisticated students to respond to this commitment by reducing demand for the school in which student $i$ is interested. We observe that this effect completely disappears when priorities are strict. Indeed, since the outcome of BM under strict priorities (given a fixed profile of students' preferences) is deterministic rather than randomized, if student $i$ is harmed by the competition of student $j$ for a certain school (in the case of strict preferences, this means that student $j$ has higher priority at this school), then student $i$ committing to apply to that school does not deter student $j$ from applying to that school, as student $j$ has priority in that school and so will not be crowded out by student~$i$.
	
	As demonstrated by \citet{ps2008} in the context of strict priorities, being sincere also has an adverse effect: not responding to excess demand for one's favorite school. As it turns out, this effect can still be manifested in markets with weak priorities. In fact,
	in some cases the sincerity of a student $i$ may not crowd out any students whatsoever, but may nonetheless harm student $i$ as she does not respond to excess demand for her favorite school. This is precisely what happens in the analysis of \cref{all-sincere-then-prefer-sophisticated} for student $z$. In that example, her sincerity does not crowd out the competition, but it does cause her to not respond to excess demand.
	
What happens when both effects of sincerity are present? Which effect dominates: the positive effect of crowding out others or the negative effect of not responding to excess demand?\footnote{A slightly different way to put things is to say that the players' decisions of whether to be sincere are strategic substitutes. If many of the other players are sincere, being sincere as well becomes less appealing, as it crowds out others less efficiently. We do not want to overemphasize this description because being sincere or sophisticated in our model is not a strategic choice, but rather the dimension along which we compare expected welfare.} In \cref{all-sincere-then-prefer-sophisticated}, the competition that student $z$ faces for school $a$ is completely decoupled from her competition for school $b$. As long as this feature is maintained, it is clear that by only tweaking the demand for the second choices of $z$'s competition for school $a$, we may change whether, and to what extent, a sincere $z$ is able to crowd out the competition, without changing the expected utility of a sophisticated $z$ from being assigned to any of the schools. This way, we may easily create variants of this market where the positive effect dominates the negative one or vice versa. A more interesting question, therefore, is what happens when the markets for school $a$ and for school $b$ are entangled? In other words, is the lack of symmetry among the different students in \cref{all-sincere-then-prefer-sophisticated} required for sincerity to be disadvantageous? To be more specific, what happens in the extreme case where all students are symmetric: which effect of sincerity dominates then? The following \lcnamecref{sincerity_bad} alludes to an answer to this question.
	
	\begin{example}[Negative Effects of Sincerity May Dominate Positive Effects]\label[example]{sincerity_bad}
		Let $\S = \{s_1,s_2,s_3\}$ and let $\I = \{i_1,i_2,\dots,i_n\}$. The schools have the following capacities:
		\begin{align*}
		q_{s_1} &= 1, \\
		q_{s_2} &= 1, \\
		q_{s_3} &= n-3.
		\end{align*}
		Every $i \in \I$ has $u_i(s_1) = 9$, $u_i(s_2) = 1$, and $u_i(s_3) = \frac{1}{2(n-3)}$. Suppose that $i_2,i_3,\dots,i_n$ are all sophisticated.
		
		If $i_1$ is also sophisticated, then in a symmetric equilibrium\footnote{A symmetric equilibrium exists since the game is finite and symmetric \citep{Nash1951}. Moreover, one can verify that when $n$ is large enough, the only symmetric equilibrium is for each student to report $s_1 \succ s_3 \succ s_2$ with some probability $t'$, and report $s_2 \succ s_3 \succ s_1$ with probability $1-t'$.} all students get the same expected utility, and since the sum of utilities is $9 + 1 + (n-3)\cdot \frac{1}{2(n-3)} = 10.5$, each student's expected utility is~$\frac{10.5}{n}$.
		
		If $i_1$ is sincere (and therefore reports $s_1\succ s_2\succ s_3$ with probability $1$), then for large enough~$n$, the only equilibrium in which $i_2,\ldots,i_n$ play symmetrically is for each of them to report $s_1 \succ s_3 \succ s_2$ with some probability~$t$, and report $s_2 \succ s_3 \succ s_1$ with probability\footnote{An equilibrium in which $i_2,\ldots,i_n$ play symmetrically exists since the game is finite and symmetric with respect to these students \citep{Nash1951}. To see why only these two strategies are played for large enough $n$, notice that an application to $s_1$ or $s_2$ in the second stage has an exponentially small probability of success, and therefore when $n$ is large enough, ranking $s_3$ second gives higher utility as it ensures admission to school $s_3$ (since $i_1$ will not apply to it before the third round).} $1-t$. In \cref{apndx:sincerity_bad}, we show that $t\ge\frac{90}{101}$ for large enough $n$, which implies that as $n$ grows large, the expected utility of $i_1$ approaches $\frac{9}{1+t(n-1)}<\frac{9}{tn}\le\frac{10.1}{n}$, and so for large enough $n$ the expected utility of $i_1$ must be strictly smaller than~$\frac{10.5}{n}$, and thus becoming sincere harms~$i_1$.
	\end{example}
	
	\cref{sincerity_bad} demonstrates that even when students share the same preferences and the same cardinal utility function, the crowding-out effect does not necessarily dominate the negative implications of being sincere and not responding to excess demand. Sincerity turns out to be a package deal, coupling together the inability to respond to excess demand with the power of crowding-out others, and in this case the negative effect outweighs the positive one. While student $i_1$, if she is sincere, slightly increases her chances of being admitted to school $s_1$, she completely forfeits her chances of being admitted to school $s_3$ by approaching school~$s_2$ in the second round; nevertheless, she does so even though in this round school $s_2$ is already very likely to be full.
	
	\subsection{The Choice of Mechanism:\texorpdfstring{\\}{ }Reduced Competition vs.\ Reduced Options}\label{effects-boston}
	
	The results of \cref{sec:bm-da} demonstrate that for a sophisticated student $i$, BM may have a negative effect as it may reduce the options of student $i$, effectively forcing her to ex-ante choose to apply to only one overdemanded school. As we noted in \cref{effects-sincerity}, this negative effect also completely disappears when priorities are strict. Indeed, in an equilibrium of BM, since the outcome is deterministic rather than randomized, student $i$ never has any reason to apply to a school she will not be (deterministically) accepted to, and therefore she is unharmed by being effectively forced to ex-ante choose to apply to only one overdemanded school: as she will never be rejected by this school, she is unharmed by giving up her ``plan B.''
	
	The positive effect of using BM is discussed by \citet{ps2008} in the context of strict priorities: reduced competition (in particular, for one's top choice), which can also be seen to manifest in markets with weak priorities. This is demonstrated by the following \lcnamecref{no-reduced-competition}, which also sheds light on the interplay between this positive effect of reduced competition and the above-discussed adverse effect of reduced options.
	
	\begin{example}[Positive Effects of BM Compared with DA May Dominate Negative Effects]\label[example]{no-reduced-competition}
		Let $S=\{s_1,s_2\}$ and for ease of presentation let the number of students $n$ be large and divisible by~$12$. Each $s \in \S$ has capacity $q_s = 1$. 
		Each student has utility $1$ for her most-preferred school, and utility $1\!-\!\varepsilon$, for very small $\varepsilon$, for her second-preferred school,
		with one quarter of the students preferring $s_1$ the most, and the remaining three quarters of the students preferring $s_2$ the most.
		Under DA, each student has a probability of roughly $\nicefrac{1}{n}$ of being admitted to each school, for a total expected utility of roughly $\nicefrac{2}{n}$.
		Under BM, if a large fraction (say, one half) of the students are sophisticated, then in equilibrium roughly half of the students will apply to each school, for an expected utility of roughly $\nicefrac{2}{n}$ for each student. If, however, only a small fraction (say, one twelfth) of the students are sophisticated, then only these students and the quarter of the students who truly prefer~$s_1$ will apply to $s_1$, resulting in lower competition (at most $\nicefrac{n}{3}$ students) for~$s_1$ than for~$s_2$ and an expected utility of strictly more than $\nicefrac{2}{n}$ for each sophisticated student (and for each sincere student who prefers $s_1$). Therefore, in this case sophisticated students strictly prefer BM to DA.
	\end{example}
	
	The question that naturally arises is which of the two effects generally dominates for a sophisticated student: reduced competition or reduced options. It seems that the answer depends on whether the ``overall competition'' was reduced.
	
	To give one example along the lines of \cref{no-reduced-competition}, consider a sophisticated student $i$ and call the top of her preference list that consists of overdemanded schools her ``overdemanded set.'' If all other students have the same overdemanded set of size $o$ (and their priorities are randomly selected), then under DA, student $i$ essentially faces competition from all the other students for the $o$ schools in her overdemanded set. Under BM, student $i$ faces competition from an average of $\nicefrac{1}{o}$ of the other students in one school and, therefore, if this competition is evenly spread (weighted by the utilities of student $i$ for the different schools), then DA and BM give similar utility to student $i$, whereas if this competition is not evenly spread then BM gives higher utility to student $i$.
	
	In contrast to the example of a shared overdemanded set, suppose that all other students have very attractive outside options,
	to the extent that each other student~$j$ only prefers one of the schools in student $i$'s overdemanded set to $j$'s outside option. In such a case, student~$i$ most certainly prefers DA, because BM not only does not reduce the competition student $i$ faces but, in fact, reduces her options. This is precisely what happens in the analysis of \cref{some-sophisticated-then-bm-da} for student~$z$. In that example, BM reduces the options of student $i$ without reducing the competition she faces.

\section{Conclusion} \label{sec:discussion}
Market designers often encounter markets that are governed by manipulable mechanisms, with BM being possibly the most prominent example. Common sense dictates that it is better to switch to a strategy-proof mechanism, as it allows the designer to directly optimize some target function (e.g., efficiency), subject to certain desirable constraints (e.g., stability), and to preserve incentive compatibility. When a manipulable mechanism is in place, it is difficult to predict what properties the likely outcome will satisfy, and whether the strategic situation will give an advantage to some populations over others.

In the specific case of BM, it intuitively seems that the use of this mechanism favors players who recognize strategic opportunities over players who do not. Indeed, as \cite{ps2008} show, in a strict-priority environment a student weakly prefers being sophisticated to being sincere, and BM weakly benefits sophisticated players (compared with DA). This constitutes a very strong argument against BM, as strategic sophistication and access to relevant information can be highly correlated with socioeconomic status and, in any case, these are definitely not the criteria based on which the matching should be determined. However, as we show in this paper, when considering a reasonable scenario in which the school district uses a coarse priority structure rather than a strict one, there are several new effects that may reverse the previous theoretical prediction. 

The question, therefore, is what should we recommend to practitioners. In \cref{sec:what_matters} we tried to convey our view that understanding the market structure is crucial for providing sound advice. That being said, it is also fair to assess that in most real-life markets, reacting to excess demand is quite straightforward in terms of strategic behavior, whereas crowding out others is more demanding as it requires a reputation for being sincere and for not shying away from competition. This suggests that, at the end of the day, while policy-makers who are focused on efficiency and absolute gains should probably adopt BM (which takes into account some students' cardinal preferences), in most markets policy-makers who are mostly interested in fairness considerations and relative gains should still favor DA. That being said, our results imply that this is not as automatic a recommendation as may have widely been believed, and should be taken with a grain of salt, as in some specific markets with special structures BM may not only benefit all players, but also provide the same or higher gain for students who are not sophisticated or not informed about the market. As such students many times come from populations from a weaker socioeconomic background, in such markets turning to BM may help keep or even increase diversity in highly coveted schools.

\clearpage

\appendix
\section{Omitted Proofs}\label{apndx:proofs}
	\begin{proof}[Proof of \cref{some-sophisticated-then-bm-da}]
		Let $u^1>u^2>0$. We first observe that if there exist three students $z,y,x$ and four schools $a,b,c,d$ such that the following conditions hold:
		\begin{enumerate}
			\item
			$\succ_{z}: a,b$,
			\item
			$\succ_{y}: a,c$,
			\item
			$\succ_{x}: b,d$,
			\item
			$a,b,c,d$ are not preferred first or second by any other student,
		\end{enumerate}
		then regardless of whether each of these three students (and, in fact, regardless of whether any student) is sophisticated or sincere, $z$ strictly prefers the DA outcome to any equilibrium of BM and $x$ strictly prefers any equilibrium of BM to the DA outcome.
		First note that regardless of whether $z$, $y$, and $x$ are sophisticated or sincere, they will rank truthfully under any equilibrium of BM even if they are sophisticated. For $y$ and $x$ this holds since their second choice is guaranteed, and for $z$ this holds as she prefers getting $a$ with probability $\nicefrac{1}{2}$ to getting $b$ with probability $\nicefrac{1}{2}$.
		Now, under DA $z$ gets utility $\nicefrac{1}{2}\cdot u^1+\nicefrac{1}{2}\cdot\frac{u^2}{2}$ (since with probability $0.5$, $z$ beats $y$ in the competition for $a$, and with the remaining probability $0.5$ it is the case that with probability $0.5$ $z$ beats $x$ in the competition for $b$)
		while under any equilibrium of BM $z$ gets utility $\frac{u^1}{2}$, which is strictly smaller. Under DA $x$ gets utility $\nicefrac{1}{2}\cdot u^1+\nicefrac{1}{2}\cdot(\nicefrac{1}{2}\cdot{u^1}+\nicefrac{1}{2}\cdot u^2)$ (since with probability $0.5$, $z$ beats $y$ in the competition for $a$ and so $x$ is the only applicant for $b$, and with the remaining probability $0.5$, $z$ competes with $x$ for $b$ and so $x$ wins and gets $b$ with probability $0.5$, and loses and gets $d$ with probability $0.5$)
		while under any equilibrium of BM $x$ gets utility $u^1$, which is strictly greater.
		
		For \cref{some-sophisticated-then-bm-da-prefer-bm}, let us calculate the probability that for any given $z\in\I$, there exist $y,x\in\I$ such that the above four conditions are satisfied. Let $a$ denote $z$'s most-preferred school and~$b$ denote her second-preferred school; this probability is given by
		\begin{multline*}
		\overbrace{\left(1-\left(\frac{n-1}{n}\right)^{n-1}\right)}^{\substack{\text{some $y\in I$ ranks $a$ first}\\\text{(let $c$ denote her}\\\text{second-preferred school)}}}\cdot
		\overbrace{\frac{n-2}{n-1}}^{\substack{\text{$y$ does not}\\\text{rank $b$ second}\\\text{(i.e., $c\ne b$)}}}\cdot
		\overbrace{\left(1-\left(\frac{n-1}{n}\right)^{n-2}\right)}^{\substack{\text{some $x\in I$ ranks $b$ first}\\\text{(let $d$ denote her}\\\text{second-preferred school)}}}\cdot
		\overbrace{\frac{n-3}{n-1}}^{\substack{\text{$x$ does not rank}\\\text{$a$/$c$ second}\\\text{(i.e., $d\notin\{a,c\}$)}}}\cdot\\*\cdot
		\underbrace{\left(\frac{n-4}{n}\cdot\frac{n-5}{n-1}\right)^{n-3}}_{\substack{\text{No other $i\in\I$ ranks}\\\text{$a$/$b$/$c$/$d$ first or second}}}
		\xrightarrow[n\to\infty]{}
		\left(1-\frac{1}{e}\right)^2\cdot\frac{1}{e^8}
		\end{multline*}
		Let $\tau$ be any constant slightly smaller than $(1-\nicefrac{1}{e})^2\cdot\nicefrac{1}{e^8}$; then we have that for large enough $n$, both the expected fraction of sophisticated students $z$ for which such $y,x\in\I$ exist and the expected fraction of sincere students $z$ for which such $y,x\in\I$ exist are at least~$\tau$, and so we have identified an appropriate set of students that satisfies \cref{some-sophisticated-then-bm-da-prefer-bm}.
		
		For \cref{some-sophisticated-then-bm-da-prefer-da}, let us calculate the probability that for any given $x\in\I$, there exist $z,y\in\I$ such that the above four conditions are satisfied. Indeed, let $b$ denote $x$'s most-preferred school and $d$ denote her second-preferred school; this probability is given by
		\begin{multline*}
		\overbrace{\left(1-\left(\frac{n-1}{n}\right)^{n-1}\right)}^{\substack{\text{some $z\in I$ ranks $b$ second}\\\text{(let $a$ denote her}\\\text{most-preferred school)}}}\cdot
		\overbrace{\frac{n-2}{n-1}}^{\substack{\text{$z$ does not}\\\text{rank $d$ first}\\\text{(i.e., $a\ne d$)}}}\cdot
		\overbrace{\left(1-\left(\frac{n-1}{n}\right)^{n-2}\right)}^{\substack{\text{some $y\in I$ ranks $a$ first}\\\text{(let $c$ denote her}\\\text{second-preferred school)}}}\cdot
		\overbrace{\frac{n-3}{n-1}}^{\substack{\text{$y$ does not rank}\\\text{$b$/$d$ second}\\\text{(i.e., $c\notin\{b,d\}$)}}}\cdot\\*\cdot
		\underbrace{\left(\frac{n-4}{n}\cdot\frac{n-5}{n-1}\right)^{n-3}}_{\substack{\text{No other $i\in\I$ ranks}\\\text{$a$/$b$/$c$/$d$ first or second}}}
		\xrightarrow[n\to\infty]{}
		\left(1-\frac{1}{e}\right)^2\cdot\frac{1}{e^8}
		\end{multline*}
		Since this is the same probability calculated above for \cref{some-sophisticated-then-bm-da-prefer-bm}, the proof of \cref{some-sophisticated-then-bm-da-prefer-da} is completed similarly to that of \cref{some-sophisticated-then-bm-da-prefer-bm}.
	\end{proof}

\section{Incomplete Information about Sophistication} \label{apndx:incomplete}
In \cref{some-sophisticated-then-advantage-and-disadvantage} we assumed that each player's probability of being sincere is $p$ independently of other students, and that the realizations of these draws (i.e., whether each student is sincere or sophisticated) are common knowledge. Furthermore, we ignored the issue of reputation (that is, the probability of being sincere) constituting a property of \emph{populations} rather than of \emph{individuals}, and just assumed that no two students belong to the same population. We have made these assumptions to ease the presentation, yet the result extends even when we relax them significantly. In this \lcnamecref{apndx:incomplete}, we present a more robust version of \cref{some-sophisticated-then-advantage-and-disadvantage} in which each student belongs to a population, each population is characterized by a probability of being sincere (and this probability is not too high), and each player's realized sophistication level is private information.

We consider the following ``population model'' for generating a market. We assume that there are $M$ populations. Each student belongs to one population independently at random, and her probability of belonging to population $\nu$ is $w_\nu \in (0,1)$; the realization of this draw for each student is common knowledge. Conditional on her being a member of population $\nu$, each student is independently at random determined to be sincere with probability $p_{\nu}$, and sophisticated otherwise; the realization of this draw is private and is known to no other student.

\begin{theorem}
\label{incomplete-info-sincere-vs-sophisticated}
Let\footnote{As with \cref{some-sophisticated-then-advantage-and-disadvantage}, it is straightforward to generalize this result to arbitrary values of $k$, and even to preferences of unbounded length.} $k=2$, let $u_1 > u_2 > \frac{2}{3}\cdot u_1 > 0$, and set $t^* = \frac{3(u_1 - u_2)}{u_1}$. For any population model where $p_\nu<t^*$ for every~$\nu$, there exists $\tau>0$ such that for any large enough $n$, in the uniform $(n; u^1,u^2)$ population model both of the following hold:	

\begin{enumerate}
\item\label[part]{incomplete-info-sincere-vs-sophisticated-prefer-sophisticated}
There exists a set of students of expected\footnote{The expectation is taken over preferences and population associations.} size at least $\tau n$, such that each student $i$ that belongs to any population $\nu$ in this set ex ante strictly prefers\footnote{That is, the strict preference is in expectation over her and others' sophistication.} any equilibrium had all the students in $\nu$ been sophisticated and had that been common knowledge\footnote{In such equilibria sophisticated students and all students in $\nu$ strategize, while all other students are truthtelling.} to any equilibrium (in which each student in $\nu$ is sincere with probability $p_\nu$ and this probability is common knowledge\footnote{In such equilibria sophisticated students strategize, while all other students (including those in $\nu$ that are not sophisticated) are truthtelling.}). Furthermore, each student in this set maintains this strict preference regardless of whether other populations' probabilities of being sincere change.
			
\item\label[part]{incomplete-info-sincere-vs-sophisticated-prefer-sincere}
For any equilibrium (in which each student that belongs to any population $\nu$ is sincere with probability $p_\nu$ and this is common knowledge, and in which only sophisticated students strategize),
there exists a set of students of expected\footnote{The expectation is taken over preferences and population associations (for any given mapping from realized preferences and population associations, to  equilibria).} size at least $\tau n$, such that each student $i$ that belongs to any population~$\nu$ in this set ex ante strictly prefers\footnote{That is, the strict preference is once again in expectation over her and others' sophistication.} any equilibrium had all students in $\nu$ been sincere and had that been common knowledge\footnote{In such equilibria sophisticated students, but not students in $\nu$, strategize, while all other students, including all students in $\nu$, are truthtelling.} to the given equilibrium. Furthermore, each student in this set maintains this strict preference regardless of whether or not any other students in this set individually have different probabilities of being sincere and regardless of whether or not other populations' probabilities of being sincere decrease.
\end{enumerate}
\end{theorem}

\begin{proof}
For \cref{incomplete-info-sincere-vs-sophisticated-prefer-sophisticated}, we can still use \cref{all-sincere-then-prefer-sophisticated}. In its proof it is of no importance whether students $y$, $x$, and $w$ are sophisticated or not, as each reports truthfully regardless. Since $z$ strictly prefers being sophisticated to being sincere, she also prefers being sophisticated to being sincere with any positive probability.

\medskip

For \cref{incomplete-info-sincere-vs-sophisticated-prefer-sincere}, we slightly modify the argument in \cref{all-sophisticated-then-prefer-sincere}. It is once again of no importance whether $x$, $w$, and $v$ are sophisticated or not, as each reports truthfully regardless. As for $z$ and $y$, we need them to belong to different populations, which we denote by $\nu_z$ and $\nu_x$ respectively (this happens with fixed probability $w_{\nu_z} \cdot w_{\nu_y}$, summed over all possible choices for distinct $\nu_z$ and $\nu_y$). The reason for having $z$ and $y$ belong to two different populations is that we do not want their sophistication levels to be interdependent when we change the sincerity probability of the population of one of them to $1$. 

In any equilibrium, within any gadget that contains players $z$, $y$, $x$, $w$, and $v$, there are theoretically five possible ways in which the sophisticated types of players $z$ and $y$ can play:
\begin{equilibria}
\item\label{eq-app-1-truthful} The sophisticated type of $z$ ranks truthfully, and the sophisticated type of $y$ ranks her second choice at the top.
\item\label{eq-app-2-truthful} The sophisticated type of $y$ ranks truthfully, and the sophisticated type of $z$ ranks her second choice at the top.
\item\label{eq-app-mix} The sophisticated types of $z$ and $y$ play a mixed strategy.
\item\label{eq-app-both-truthful} The sophisticated types of $z$ and $y$ rank truthfully.
\item\label{eq-app-both-untruthful} The sophisticated types of $z$ and $y$ rank their second choice at the top.
\end{equilibria}

\cref{eq-app-both-truthful} is ruled out by the restriction that $u_2 > \frac{2}{3} u_1$. \cref{eq-app-both-untruthful} is ruled out by our assumption that $p_\nu < t^*$. Indeed, if the sophisticated types of the two players had reported their second choice first, then the sophisticated type of player $z$ would have received utility $\frac{1}{2}u_2$, while by reporting truthfully they could have received utility
\begin{equation*}
p_{\nu_y} \cdot \frac{1}{3}u_1 + (1-p_{\nu_y}) \cdot \frac{1}{2}u_1 = \left(\frac{1}{2} - \frac{1}{6} p_{\nu_y}\right)\cdot u_1  > \left(\frac{1}{2} - \frac{1}{6} \cdot \frac{3(u_1-u_2)}{u_1} \right)\cdot u_1  = \frac{1}{2} u_2.
\end{equation*}

Under \cref{eq-app-1-truthful,eq-app-2-truthful}, either $z$ or $y$ strictly prefers to become sincere. Indeed, whichever one of the two currently reports her second choice first gets utility $\frac{1}{2}u_2$, and by becoming sincere she will made the sophisticated type of the other player switch to reporting her own second choice first, and will then get utility $\left(\frac{1}{2} - \frac{1}{6} p_{\nu_y}\right)\cdot u_1$, which is again higher than~$\frac{1}{2}u_2$.

Finally, we turn to analyzing \cref{eq-app-mix}. The mixed equilibrium in this case is for the sophisticated type of each player $i \in \{z,y\}$ to report truthfully with probability $\frac{t^* - p_{\nu_i}}{1 - p_{\nu_i}}$. This causes each player to face a truthful opponent with probability $t^*$, which leaves her indifferent between reporting truthfully and reporting her second choice first. In particular, this also means that the sincere type of each player gets the same utility as her sophisticated type, and so both of these types get utility $\frac{1}{2}u_2$. Therefore, as in \cref{eq-app-1-truthful,eq-app-2-truthful}, $z$ and $y$ strictly prefer to become sincere.
\end{proof}

\section{Calculation for Example~\ref{sincerity_bad}} \label{apndx:sincerity_bad}

In this \lcnamecref{apndx:sincerity_bad}, we will estimate $p$, as defined in \cref{sincerity_bad}. Thus, we assume that $i_1$ is sincere (and so reports $s_1\succ s_2\succ s_3$), and focus on analyzing the utility of $i_2$.
We first claim that for large enough $n$, it holds that $p\ge\nicefrac{1}{2}$. Indeed, if it were the case that $p<\nicefrac{1}{2}$ for such $n$, then the expected utility of $u_2$ from listing $s_1$ first would be strictly higher than her expected utility from listing $s_2$ first---a contradiction. Fix $m=n-2$ and $\delta=\nicefrac{1}{100}$. Applying a Chernoff bound \cite[see, e.g.,][p.~67]{MitzenmacherUpfalBook}, we have that the number $a_1$ of students among  $i_3,\ldots,i_n$ who rank $s_1$ first satisfies
\begin{equation}\label{chernoff}
\Pr\bigl[a_1\ge (1+\delta)\cdot mp\bigr]\le \exp\left(-\frac{\delta^2}{3}\cdot pm\right) \le \exp\left(-\frac{\delta^2}{6}\cdot m\right) = \exp\left(-\alpha\cdot m\right),
\end{equation}
where we have set $\alpha=\nicefrac{-\delta^2}{6}$.

We begin by estimating the utility of $i_2$ from reporting $s_1\succ s_3\succ s_2$. By \cref{chernoff}, this expected utility is \emph{at least}
\begin{multline*}
e^{-\alpha m}\cdot 0 + (1-e^{-\alpha m})\cdot\left(\vphantom{\left(\frac{1}{\delta p}\right)}\right.\hspace{-1.4em}\overbrace{\frac{1}{2+m\cdot(1+\delta)p}}^{\substack{\text{Probability that $i_2$ is admitted}\\\text{to $s_1$ when $m\cdot(1+\delta)p$}\\\text{sophisticated students plus $i_1$}\\\text{also apply in addition to $i_2$}}}\left.\hspace{-1.4em}\cdot9+\left(1-\frac{1}{2+m\cdot(1+\delta)p}\right)\cdot\frac{n-3}{n-2}\cdot\frac{1}{2(n-3)}\right)
\ge\\\ge
\frac{9}{2+m\cdot(1+\delta)p}+\left(1-\frac{2}{m}\right)\cdot\left(1-\frac{1}{m}\right)\cdot\frac{1}{2(m-1)}-9e^{-\alpha m}
\ge\\\ge
\frac{9}{2+m\cdot(1+\delta)p}+\left(1-\frac{3}{m}\right)\cdot\frac{1}{2(m-1)}-9e^{-\alpha m}.
\end{multline*}
We now estimate the utility of $i_2$ from reporting $s_2\succ s_3\succ s_1$. By \cref{chernoff}, this expected utility is \emph{at most}
\begin{multline*}
e^{-\alpha m}\cdot 1 + (1-e^{-\alpha m})\cdot\left(\vphantom{\left(\frac{1}{\delta p}\right)}\right.\hspace{-.2em}\overbrace{\frac{1}{1+m\cdot\bigl(1-(1+\delta)p\bigr)}}^{\substack{\text{Probability that $i_2$ is admitted}\\\text{to $s_2$ when $m\cdot\bigl(1-(1+\delta)p\bigr)$}\\\text{sophisticated students}\\\text{also apply in addition to $i_2$}}}\left.\hspace{-.2em}+\!\left(1-\frac{1}{1+m\cdot\bigl(1-(1+\delta)p\bigr)}\right)\!\cdot\frac{1}{2(n-3)}\right)
 \le\\\le
 e^{-\alpha m} + \frac{1}{1+m\cdot\bigl(1-(1+\delta)p\bigr)}+
 \frac{1}{2(m-1)}.
 \end{multline*}

As these two expected utilities are equal (since $p$ is the mixing probability in equilibrium), combining the two estimates above we get
\[ \frac{9}{2+m\cdot(1+\delta)p} \le 10e^{-\alpha m}+\frac{1}{1+m\cdot\bigl(1-(1+\delta)p\bigr)}+
\frac{3}{2m(m-1)}
. \]
Therefore,
\begin{multline*}
9+9m\cdot\bigl(1-(1+\delta)p\bigr)\le\\
\bigl(2+m\cdot(1+\delta)p\bigr)\cdot\Bigl(1+m\cdot\bigl(1-(1+\delta)p\bigr)\Bigr)\cdot\left(10e^{-\alpha m}+\frac{3}{2m(m-1)}\right) + 2+m\cdot(1+\delta)p
\le\\\le
10(m+2)(m+1)\cdot e^{-\alpha m} + 4+m\cdot(1+\delta)p
\end{multline*}
for $n$ large enough (and therefore $m$ large enough) such that $\frac{(m+2)(m+1)}{m(m-1)}\le\nicefrac{4}{3}$. Therefore,
\[5-10(m+2)(m+1)\cdot e^{-\alpha m} + 9m\le10mp\cdot(1+\delta),\]
and so, for $n$ large enough (and therefore $m$ large enough) such that $10(m+2)(m+1)\cdot e^{\alpha m}\le5$, we have that
\[9m\le10mp\cdot(1+\delta).\]
Therefore,
\[\frac{9}{10\cdot(1+\delta)}\le p,\]
and substituting $\delta=\nicefrac{1}{100}$, we obtain that for large enough $n$ that $p\ge\frac{90}{101}$, as claimed.\qed

\clearpage

\bibliographystyle{plainnat}
\bibliography{boston-weak-priorities}

\end{document}